\documentclass[final,1p]{elsarticle}

\usepackage{silence}
\ActivateWarningFilters[pdftoc]

\usepackage[utf8]{inputenc}
\usepackage[T1]{fontenc}
\usepackage{enumerate}
\usepackage{amsfonts}
\usepackage{amssymb}
\usepackage{amsmath}
\usepackage{amsthm}
\usepackage{wasysym}
\usepackage[misc]{ifsym}
\usepackage{mathtools}
\usepackage{comment}
\usepackage{paralist}
\usepackage[inline]{enumitem}
\usepackage{todonotes}
\usepackage{xspace}

\usepackage{xcolor}
\definecolor{defblue}{rgb}{0.121,0.47,0.705}
\definecolor{linkblue}{rgb}{0.098,0.098,0.5392}
\definecolor{blue}{rgb}{0.098,0.098,0.5392}
\let\emph\relax
\DeclareTextFontCommand{\emph}{\color{defblue}\em}
\DeclareTextFontCommand{\bl}{\color{defblue}}
\usepackage[colorlinks=true,
    linkcolor=linkblue,
    anchorcolor=linkblue,
    citecolor=linkblue,
    filecolor=linkblue,
    menucolor=linkblue,
    urlcolor=linkblue,
    bookmarks=true,
    bookmarksopen=true,
    bookmarksopenlevel=2,
    bookmarksnumbered=true,
    hyperindex=true,
    plainpages=false,
    pdfpagelabels=true
]{hyperref}
\usepackage[capitalise,noabbrev,nameinlink,sort]{cleveref}
\usepackage{subcaption}

\usepackage{thmtools}
\usepackage{float}

\newtheorem{claim}{Claim}{\itshape}{\rmfamily}
\crefname{claim}{Claim}{Claims}
\Crefname{claim}{Claim}{Claims}
{\bfseries}{\itshape}
\crefname{observation}{Observation}{Observations}
\Crefname{observation}{Observation}{Observations}
\newtheorem{corollary}{Corollary}{\bfseries}{\itshape}
\crefname{corollary}{Corollary}{Corollaries}
\Crefname{corollary}{Corollary}{Corollaries}
\newtheorem{lemma}{Lemma}{\bfseries}{\itshape}
\crefname{lemma}{Lemma}{Lemmas}
\Crefname{lemma}{Lemma}{Lemmas}
\newtheorem{theorem}{Theorem}{\bfseries}{\itshape}
\crefname{theorem}{Theorem}{Theorems}
\Crefname{theorem}{Theorem}{Theorems}
\newtheorem{proposition}{Proposition}{\bfseries}{\itshape}
\crefname{proposition}{Proposition}{Propositions}
\Crefname{proposition}{Proposition}{Propositions}
\crefname{condition}{Condition}{Conditions}
\Crefname{condition}{Condition}{Conditions}

\DeclareMathOperator{\ch}{CH}
\newcommand{\vis}{\ensuremath{\mathcal{V}}}
\def\pch{\partial\ch}
\newcommand{\bsp}{balanced separated partition\xspace}

\graphicspath{{./figures/}}%

\newcounter{casecounter}
\newcounter{subcasecounter}
\renewcommand{\thesubcasecounter}{\thecasecounter.\arabic{subcasecounter}}
\newcounter{subsubcasecounter}
\renewcommand{\thesubsubcasecounter}{\thesubcasecounter.\arabic{subsubcasecounter}}
\makeatletter
\newcommand{\ccase}[1]{%
  \refstepcounter{casecounter}%
  \setcounter{subcasecounter}{0}%
  \setcounter{subsubcasecounter}{0}%
  \medskip\noindent\textbf{Case \thecasecounter.}
  \label{#1}
}
\newcommand{\subcase}[1]{%
  \refstepcounter{subcasecounter}%
  \setcounter{subsubcasecounter}{0}%
  \noindent\textbf{Case~\thesubcasecounter.}
  \label{#1}
}
\newcommand{\subsubcase}[1]{%
  \refstepcounter{subsubcasecounter}%
  \noindent\textbf{Case~\thesubsubcasecounter.}
  \label{#1}
}
\makeatother
\crefname{casecounter}{Case}{Cases}
\crefname{subcasecounter}{Case}{Cases}
\crefname{subsubcasecounter}{Case}{Cases}
\Crefname{casecounter}{Case}{Cases}
\Crefname{subcasecounter}{Case}{Cases}
\Crefname{subsubcasecounter}{Case}{Cases}

\journal{Discrete Mathematics}

\begin{document}

\begin{frontmatter}

\title{Three Edge-disjoint Plane Spanning Paths in a Point Set}

\author[trier]{Philipp Kindermann}
\ead{kindermann@uni-trier.de}
\address[trier]{FB IV - Computer Science, Trier University, Trier, Germany}

\author[prague]{Jan Kratochvíl}
\ead{honza@kam.mff.cuni.cz}
\address[prague]{Department of Applied Mathematics, Faculty of Mathematics and Physics, Charles University, Prague, Czech Republic}

\author[perugia]{Giuseppe Liotta}
\ead{giuseppe.liotta@unipg.it}
\address[perugia]{Department of Engineering, University of Perugia, Perugia, Italy}

\author[prague]{Pavel Valtr}
\ead{valtr@kam.mff.cuni.cz}

\begin{abstract}
We consider the following problem: Given a set $S$ of $n$ distinct points in the plane, how many edge-disjoint plane straight-line spanning paths can be drawn on $S$? Each spanning path must be crossing-free, but edges from different paths are allowed to intersect at arbitrary points. It is known that if the points of $S$ are in convex position, then $\lfloor n/2 \rfloor$ such paths always exist. However, for general point sets, the best known construction yields only two edge-disjoint plane spanning paths.

In this paper, we prove that for any set $S$ of at least ten points in general position (i.e., no three points are collinear), it is always possible to draw at least three edge-disjoint plane straight-line spanning paths. Our proof relies on a structural result about halving lines in point sets and builds on the known two-path construction, which we also strengthen: we show that for any set $S$ of at least six points, and for any two specified points on the boundary of the convex hull of $S$, there exist two edge-disjoint plane spanning paths that start at those prescribed points.

Finally, we complement our positive results with a lower bound: for every $n \geq 6$, there exists a set of $n$ points for which no more than $\lceil n/3 \rceil$ edge-disjoint plane spanning paths are possible.
\end{abstract}


\end{frontmatter}

\section{Introduction}\label{sec:Intro}

Let $S$ be a set of distinct points (locations) in the plane. We want to compute three \emph{edge-disjoint} spanning paths of $S$. Note that the edges of each path are straight-line segments and that no two edges of a same path can cross while the edges of distinct spanning paths may cross (i.e. share points that are not elements of $S$).

At a first glance, one would be tempted to generalize the question and  study the existence of $k$ such edge-disjoint spanning paths with $k\geq 2$. Namely, the proof of Bernhart and Kainen about the book thickness of a complete graph \cite[Theorem 3.4]{DBLP:journals/jct/BernhartK79} already gives a partial answer: if the $n$ points are in convex position, then it is possible to draw  $\lfloor \frac{n}{2} \rfloor$ edge-disjoint plane straight-line spanning paths of the point set which is also a tight upper bound for even values of $n$ (the complete graph has $\frac{n(n-1)}{2}$ edges). However, little is known when the $n$ points are not in convex position: The only result we are aware of is by Aichholzer et al.~\cite{ahkklpsw-ppstpcgg-17}, who show the existence of two edge-disjoint plane straight-line spanning paths for any set of $n\geq 4$ points in general position (no three collinear). Aichholzer et al.\ leave %
open the problem of proving whether three or more paths always exist. Our main result is as follows.

\begin{restatable}{theorem}{ThmMain}
\label{thm:main}
Let $S$ be any set of at least ten points in general position in the plane.  There are three edge-disjoint plane straight-line spanning paths of $S$.
\end{restatable}

Besides addressing an open problem by Aichholzer et al.~\cite{ahkklpsw-ppstpcgg-17}, 
\cref{thm:main} relates with some classical topics in the 
graph drawing literature. Among them, the \emph{graph packing problem} asks 
whether it is possible to map a set of smaller graphs into a larger graph, 
called the \emph{host graph}, without using the same edge of the host graph 
twice. A rich body of literature is devoted to this problem, both when the 
host graph is the complete graph and when the smaller graph is either planar 
or near-planar~\cite{Bollobas1978,DBLP:journals/jocg/GeyerHKKT17,DBLP:journals/ajc/HalerW14,Hedetniemi1982,DBLP:journals/jgaa/LucaGHKLLMTW21,Teo1990}.
While most papers devoted to graph packing  do not assume that a drawing of 
the host graph is given as part of the input, our study considers a 
\emph{geometric graph packing problem}, as we want to map three plane geometric 
paths with $n$ vertices into a complete geometric graph $K_n$. Bose et 
al.~\cite{DBLP:journals/comgeo/BoseHRW06} give a characterization of those plane 
trees that can be packed in a complete geometric graph $K_n$ in the special 
case that the vertices of $K_n$ are in convex position. Aichholzer et 
al.~\cite{ahkklpsw-ppstpcgg-17} show that $\Omega(\sqrt{n})$ edge-disjoint plane 
trees can be packed into a complete geometric graph with $n$ vertices, but it 
is not known whether this lower bound extends to paths. Biniaz et al.~\cite{DBLP:journals/dmtcs/BiniazBMS15} 
show that every set of $n$ points in general 
position admits $\lceil \log_2 n \rceil -1$ edge disjoint perfect matchings 
and that there exist point sets for which the maximum number of  edge 
disjoint perfect matchings is at most $\lceil \frac{n}{3} \rceil$.  From the 
perspective of geometric graph packing problems, \cref{thm:main} directly 
implies the following.

\begin{corollary}\label{thm:main-alt}
Three edge-disjoint plane Hamiltonian paths can be packed into every complete geometric graph with at least ten vertices.
\end{corollary}

As an additional result, we prove that for any $n \geq 6$, there exists a point set of size $n$ for which the maximum number of edge-disjoint plane spanning paths is at most $\lceil \frac{n}{3} \rceil$.

\paragraph{Paper organization} The rest of the paper is organized as follows. In \cref{sec:prelim}, we introduce technical notions and briefly recall the concept of zig-zag paths of Abellanas et al.~\cite{DBLP:journals/dam/AbellanasGHNR99}, and then in \cref{sec:Overview}, we present an overview of the  proof of our main theorem. \cref{sec:Twopath} is devoted to a strengthening of the result of Aichholzer et al.~\cite{ahkklpsw-ppstpcgg-17} on the existence of two paths even when the starting points are prescribed, which we believe is also interesting on its own. \cref{sec:Main} then presents the proof of our main theorem in detail. \Cref{se:Upper} contains the proof of the upper bound, while a final discussion with open problems can be found  in \cref{se:conclusion}.%

\section{Preliminaries}\label{sec:prelim}

We denote a path in a graph by its sequence of vertices and edges. For a path $P$ starting at $v_1$, ending at $v_h$, and visiting its vertices $v_i$ in increasing order of the subscript $i$ from $1$ to $h$, we represent the sequence as $P = v_1 \circ v_1v_2 \circ v_2\circ \cdots \circ v_{h-1}v_h \circ v_h$. Thus, the concatenation of vertex-disjoint paths $P$ (ending with vertex $x$) and $Q$ (starting with vertex $y$ adjacent to $x$) is the path $P\circ xy \circ Q$.  By $P^{-1}$ we denote the path $P$ traversed in the reversed order.

Let $S$ be  a set of points in general position (i.e., no three collinear) in the plane and let $p,q$ be two points of $S$.
We denote by $\overline{pq}$ the line passing through them, and by ${pq}$ its line segment with endpoints $p$ and $q$.
The \emph{geometric graph} $G(S)$ determined by $S$ is the straight-line drawing of the complete graph with vertex set $S$. A \emph{path on} $S$ is a straight-line drawing of a path in $G(S)$. For the sake of brevity we shall omit the term ``straight-line''.
A path on $S$ is \emph{spanning} if its set of vertices is the entire set~$S$.  We shall say \emph{plane (spanning) path} to mean that the (spanning) path is crossing-free.

We denote by $\ch(S)$ the convex hull of $S$ and by  $\pch(S)$ the boundary of its convex hull. Points of $S\cap\pch(S)$ are the {\emph{extreme points}} of $S$. We will call a partition $S=S_1\cup S_2$ a \emph{\bsp} if the two sets are almost equal in sizes (i.e., $||S_1|-|S_2||\le 1$) and $\ch(S_1)\cap \ch(S_2)=\emptyset$. In such a case we denote the partition as $(S_1,S_2)$. The boundary of $\ch(S)$ contains two edges with one end-point in $S_1$ and the other one in $S_2$; such edges are called \emph{bridges} of the partition, and each vertex incident with a bridge is called \emph{bridged}. A line $\ell$ is a \emph{balancing line} for a set $S$ if the intersections of $S$ with the open half-planes determined by $\ell$ form  a \bsp\ of $S\setminus\ell$. Note that a balancing line for $S$ may contain 0, 1 or 2 points of $S$. Every point of $ S$ belongs to at least one balancing line passing through
this point, but not every two points of $S$ belong to the same balancing line. However, every set (of size at least 2) contains two points which determine a balancing line.

If $(S_1,S_2)$ is a \bsp of $S$, a \emph{zig-zag $(S_1,S_2)$-path} is a plane spanning path in $S$ in which the points from $S_1$ and $S_2$ alternate. When $S_1$ and $S_2$ are clear from the context, we just call it a zig-zag path. It is well known that a zig-zag path exists for every \bsp {} of $S$~\cite{DBLP:journals/dam/AbellanasGHNR99,DBLP:journals/bit/HershbergerS92}.

\begin{lemma}[\cite{DBLP:journals/dam/AbellanasGHNR99}]\label{lem:zigzag}
	Every \bsp  admits a zig-zag path.
\end{lemma}

The algorithm to build a zig-zag path by Abellanas et al.~\cite{DBLP:journals/dam/AbellanasGHNR99} works roughly as follows.
Assume that $(S_1,S_2)$ is a partition of $S$, $S_1$ and $S_2$ are separated by a horizontal line and that $|S_1|\ge |S_2|$. Let $(p_1,q_1)$ be the \emph{left bridge} of $S$, i.e., the left edge of $\ch(S)$ crossing the separating line, with $p_1\in S_1$ and $q_1\in S_2$. Start with $P=p_1$, let $V(P)$ the set of points in $P\cup S$, and let $r$ be the last point of $P$. Inductively compute the left bridge $pq$ of $S\setminus V(P)$, and set $P=P\circ rp\circ p$ if $r$ in $S_2$, or set $P=P\circ rq\circ q$ if $r$ was in $S_1$. Continue this process until all vertices of $S$ are added to $P$. Then $P$ is a zig-zag path. Note here that if $|S_1|=|S_2|$, we may choose if the zig-zag path starts in $S_1$ or in $S_2$, but when the sets $S_1$ and $S_2$ are not equal in sizes, every zig-zag path must start in the bigger one of them. And that on a zig-zag path constructed in the above sketched way, the crossing points of the edges of the path with the separating horizontal line have the same linear order along the path and along the separating line.

\section{Approach overview}\label{sec:Overview}

The main idea of our approach is to construct three edge-disjoint plane spanning paths $Z,P,Q$ on a point set $S$ with $|S|\ge 10$ as follows. We find a suitable \bsp $(S_1,S_2)$ of $S$. The first path $Z$ will be the zig-zag path obtained by \cref{lem:zigzag}. For $P$ and $Q$, we seek to find two edge-disjoint plane spanning paths $P_1,Q_1$ (ending in $p_1$ and $q_1$, respectively) in $S_1$ and two  edge-disjoint plane spanning  paths $P_2,Q_2$ (starting in $p_2$ and $q_2$, respectively) in $S_2$; these are obviously edge-disjoint with $Z$. If $p_1p_2$ and $q_1q_2$ do not belong to $Z$ and their interiors are disjoint with $\ch(S_1)\cup\ch(S_2)$, we combine these four paths to two edge-disjoint plane spanning paths $P=P_1\circ p_1p_2 \circ P_2$ and $Q_1\circ q_1q_2\circ Q_2$ in $S$. To this end, we employ the strategy in the opposite order: we start with finding two pairs of vertices on $\partial\ch(S_1)$ and $\partial\ch(S_2)$ that see each other (i.e., their connections do not go through $\ch(S_1)\cup\ch(S_2)$) and that are not connected by an edge in $Z$, and then find spanning trees in $S_1$ and $S_2$ that start in these vertices. See also \cref{fig:overview} for a schematic description of our approach.

\begin{figure}[t!]
\centering
\subcaptionbox{$Z$\label{fig:overview-P1}}{\includegraphics[page=2]{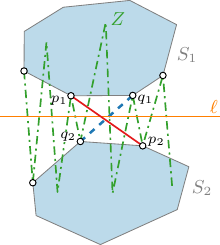}}
\hfill
\subcaptionbox{$P$\label{fig:overview-P2}}{\includegraphics[page=3]{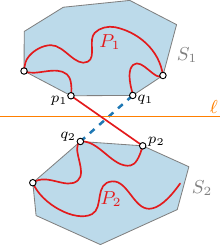}}
\hfill
\subcaptionbox{$Q$\label{fig:overview-P3}}{\includegraphics[page=4]{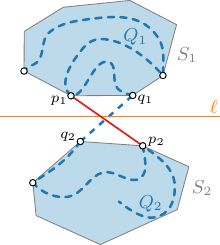}}
\caption{Schematic illustration of the approach behind the proof of \cref{thm:main}.}
\label{fig:overview}
\end{figure}

There are several difficulties: First, we cannot use the algorithm by Aichholzer et al.~\cite{ahkklpsw-ppstpcgg-17} to find the two paths in $S_1$ and $S_2$, as that does not control the starting and ending points of the paths. %
Hence, we strengthen their theorem and prove that one can always find two edge-disjoint plane spanning paths even if their starting points  are prescribed (\cref{thm:pq-paths}). Secondly, it might not be possible to find two pairs of vertices on the convex hulls with our desired properties. However, we prove that in most of such situations, the zig-zag path can be slightly modified so that the connection is possible. Lastly, we show that all of the previously described moves fail in one and only one very  specific configuration, which allows  three edge-disjoint plane spanning paths to be constructed easily in an ad hoc way, thus establishing \cref{thm:main}.

\section{Two edge-disjoint plane spanning paths with prescribed starting points}\label{sec:Twopath}

Let $S$ be a set of at least five distinct points in the plane in general position and let $s$ and $t$ be two distinguished elements of $S$, possibly coincident. In this section we show that there exist two edge-disjoint plane spanning paths of $S$, one starting at $s$ and the other one starting at $t$ such that $st$ is not an edge of either path. To this aim, we start with some basic properties of planar point sets.

Let $p$ be a point outside $\ch(S)$. We say that $p$ \emph{sees} a point $q\in S$ if ${pq}\cap \ch(S) = \{q\}$.
We denote by $S(p)$ the set of (extreme) points of $S$ that are seen from $p$.

The following lemma follows immediately from a radial sweep around $p$; see \cref{fig:point-visibility}.

\begin{lemma}
\label{lem:seestwo}
Let $|S|\ge 3$ and let $p$ be a point outside $\ch(S)$. Then $p$ sees at least two points of $S$.
	Moreover, $S(p)$ forms a continuous interval in $S\cap\pch(S)$ (along $\pch(S)$).
\end{lemma}

\begin{figure}[t]
\centering
\subcaptionbox{\cref{lem:seestwo}\label{fig:point-visibility}}{\includegraphics{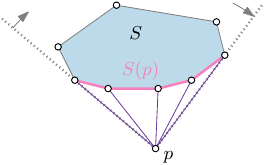}}
\hfil
\subcaptionbox{\cref{lem:claimbad2case}\label{fig:bad2case}}{\includegraphics{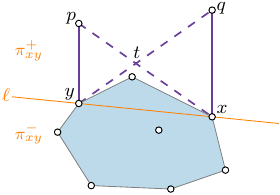}}
\caption{Illustration for the proof of (\subref{fig:point-visibility}) \cref{lem:seestwo} and (\subref{fig:bad2case}) \cref{lem:claimbad2case}.}
\end{figure}

\begin{lemma}
\label{lem:claimbad2case}
Let $|S|\ge 3$ %
	and let $p,q$ be 2 distinct points outside $\ch(S)$ such that $S\cup\{p,q\}$ is in general position; let $x$ and $y$ be two extreme points of $S$. If $x\in S(p), y\in S(q)$ and the (visibility) segments $px$ and $qy$ cross in an interior point,  then $\{x,y\}\subseteq S(p)\cap S(q)$ and $py\cap qx=\emptyset$.
\end{lemma}

\begin{proof}
Let $\ell=\overline{xy}$, let $\pi^+_{xy}$ be the half-plane determined by $\ell$ that contains the crossing point $t$ of $px$ and $qy$, and let $\pi^-_{xy}$ be the opposite half-plane determined by $\ell$; see \cref{fig:bad2case}. Let $\Delta xyt$ be the region bounded by the triangle $xyt$. Clearly $p,q\in \pi^+_{xy}$. Then $S\subseteq \pi^-_{xy}\cup \Delta xyt$, hence also $\ch(S)\subseteq \pi^-_{xy}\cup \Delta xyt$ (here we are using the fact that both $x$ and $y$ lie on $\pch(S)$), and thus $py\cap \ch(S)=\{y\}$ and $qx\cap \ch(S)=\{x\}$. The segments $py$ and $qx$ are non-crossing, since the edges of a complete graph on 4 vertices ($x,y,p,q$) may cross in at most one point, and this crossing is already consumed by $t$.
\end{proof}

\begin{lemma}
\label{lem:bad2case}
Let $|S|\ge 3$ %
and let $p,q$ be two distinct points outside $\ch(S)$ such that $S\cup\{p,q\}$ is in general position. %
Assume $|S(p)\cup S(q)|\ge 3$. Then for any point $c\in S(q)$, there exist points $a\in S(p)$ and $b\in S(q)\setminus\{c\}$ such that $ap\cap bq = \emptyset$.
\end{lemma}

\begin{proof}
If each of $p$ and $q$ sees exactly 2 points of $S$ which lie on $\pch(S)$, then they must see different pairs of points and no two of the visibility segments cross, and the statement is clear; see \cref{fig:bad2case2-22}.

Suppose at least one of $p$, $q$ sees at least 3 points of $S$ which lie on $\pch(S)$.

If $|S(q)|=2$, say $S(q)=\{c,d\}$, then we choose $b=d$ and as $a\in S(p)$, we choose a point different from $c$ and $d$ (there must be at least one, since $|S(p)|\ge 3$); see \cref{fig:bad2case2-23} Then $pa\cap qb = \emptyset$, since otherwise $a\in S(q)$ by \cref{lem:claimbad2case}, and $|S(q)|\ge 3$.

If $|S(q)|\ge 3$, we set $x$ to be a point in $S(p)\setminus \{c\}$ and $y$ to be a point in $S(q)\setminus\{x,c\}$; see \cref{fig:bad2case2-33}. By \cref{lem:seestwo}, $|S(p)|\ge 2$, and thus the points $x$ and $y$ exist, they are distinct and both of them are distinct from $c$. If $px\cap qy=\emptyset$, we set $a=x$ and $b=y$. If $px$ and $qy$ cross, then \cref{lem:claimbad2case} implies that $y\in S(p), x\in S(q)$ and $py$ does not cross $qx$. Hence, we set $a=y$ and $b=x$.
\end{proof}

\begin{figure}[t]
\centering
\subcaptionbox{$|S(p)|=|S(q)|=2$\label{fig:bad2case2-22}}{\includegraphics[page=2]{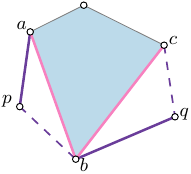}}
\hfil
\subcaptionbox{$|S(p)|\ge 3$,$|S(q)|=2$\label{fig:bad2case2-23}}[.35\linewidth]{\includegraphics[page=3]{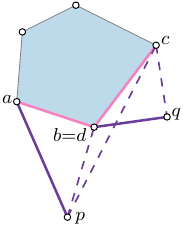}}
\hfil
\subcaptionbox{$|S(q)\ge 3$\label{fig:bad2case2-33}}{\includegraphics[page=4]{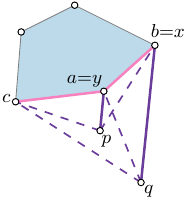}}
\caption{Illustration the proof of \cref{lem:bad2case}.}
\label{fig:bad2case2}
\end{figure}

\begin{lemma}\label{lem:pq-path}
Let $s$ and $t$ be two distinct points  of $S$. Then $S$ contains a plane spanning path which starts at $s$ and ends at $t$.
\end{lemma}

\begin{proof} Suppose first that both $s$ and $t$ lie on $\pch(S)$; see \cref{fig:pq-path-boundary}.
Let $\ell_s$ ($\ell_t$) be a supporting line of $S$ passing through $s$ (through $t$) and let $x$ be the crossing point of $\ell_s$ and $\ell_t$ (since the points of $S$ are in general position, we may assume without loss of generality that $\ell_s$ and $\ell_t$ are not parallel and that $S\cup\{x\}$ is in general position). %
Consider the lines $\overline{xy}$, for $y\in S$, and order them $\ell_1, \ell_2, \ldots, \ell_{|S|}$ as they form a rotation scheme around $x$ from $\ell_1=\ell_s$ to $\ell_{|S|}=\ell_t$. Rename the points of $S$ as $y_i\in \ell_i$, $i=1,2,\ldots,|S|$. Then $s{=}y_1 \circ y_1y_2 \circ y_2 \circ \ldots \circ  y_{|S|}{=}t$ is a plane spanning path starting at $s$ and ending at $t$.

	Now assume that at least one of $s,t$ is an interior point of $\ch(S)$, say $s$;
see \cref{fig:pq-path-inside}. The line $\overline{st}$ separates $S\setminus\{s,t\}$ into two disjoint nonempty sets $A,B$. %
Also, this line intersects the relative interior of an edge of $\pch(S)$, say $ab$ with $a\in A$ and $b\in B$. %
Now both $s$ and $a$ lie on $\pch(A\cup \{s\})$, and the previously proven case implies existence of a plane spanning path $P_A$ in $A\cup\{s\}$ which starts in $s$ and ends in $a$. Similarly, $B\cup\{t\}$ contains a plane spanning path $P_B$ which starts in $b$ and ends in $t$. Then $P_A \circ ab \circ P_B$ is the desired path.
\end{proof}

\begin{figure}[t]
\centering
\subcaptionbox{$s,t\in\pch(S)$\label{fig:pq-path-boundary}}{\includegraphics[page=1]{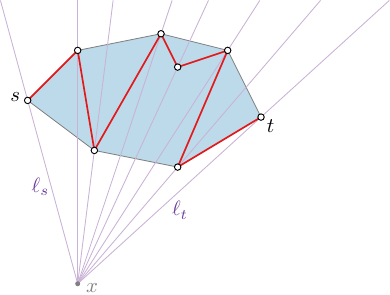}}
\hfill
\subcaptionbox{$s,t\notin\pch(S)$\label{fig:pq-path-inside}}{\includegraphics[page=2]{pq-path.pdf}}
\caption{Illustration for the proof of \cref{lem:pq-path}.}
\label{fig:pq-path}
\end{figure}

The following result will be used in the proof of \cref{thm:main}, but it also provides a strengthening of the result by Aichholzer et al.~\cite{ahkklpsw-ppstpcgg-17} on the existence of two edge-disjoint plane straight-line spanning paths in a point set. %

\begin{theorem}
\label{thm:pq-paths}
Let $|S|\ge 5$ and let $s$ and $t$ be two (not necessarily distinct) points of $\partial \ch(S)$. Then $S$ contains two edge-disjoint plane spanning paths, one starting at $s$ and the other one at $t$. Moreover, if the points $s$ and $t$ are distinct, then the paths can be chosen so that none of them contains the edge ${st}$.
\end{theorem}
\begin{proof}
We distinguish between two cases, $s\neq t$ and $s=t$.

\ccase{c:unequal}\textit{$s\neq t$.} %

\subcase{sc:odd-not-halfing}\textit{$|S|$ is odd or $\overline{st}$ is not a balancing line; see \cref{fig:pq-paths-odd}.} Let $\ell$ be a balancing line passing through $s$ and no other point of $S$, chosen so that $t$ belongs to the smaller part in case $|S|$ is even. This line defines a \bsp $S\setminus\{s\}=(S_1, S_2)$. If $|S|$ is odd, the choice of the partition is unique. Without loss of generality assume that $t\in S_2$.
Since $|S|\ge 5$, we have $|S_1|\ge |S_2|\ge 2$.

Let $s_0t_0$ be the edge of $\pch(S_1\cup S_2)$ which intersects $\ell$ and is seen from the point $s$, with $s_0\in S_1$ and $t_0\in S_2$.
Let $Z$ be the zig-zag path for $S_1\cup S_2$ starting in point $s_0$. Set $P=s \circ ss_0 \circ Z$. Since the zig-zag path $Z$ lies inside $\ch(S_1\cup S_2)$ and $ss_0$ is outside it, $P$ is a non-crossing path, and it visits all points of $S$.

Let $s_1\neq s_0$ be a point on $\pch(S_1)$ which is seen from $s$ (since $|S_1|\ge 2$ and $\{s\}\cup S_1$ is in general position, it follows from \cref{lem:seestwo} that $s$ sees at least two points of $\pch(S_1)$, and thus $s_1$ exists). Similarly, let $t_1\neq t$ be a point on $\pch(S_2)$ which is seen from $s$. Let $Q_1$ be a plane spanning path for $S_2$ that starts in $t$ and ends in $t_1$ (its existence is guaranteed by \cref{lem:pq-path}), and let $P_1$ be a plane spanning path for $S_1$ starting in $s_1$ (again, such a path exists because of \cref{lem:pq-path}, we could even prescribe its ending point, but we do not bother). Set $Q= Q_1 \circ t_1s \circ s \circ ss_1 \circ P_1$. This path is plane and it visits all points of $S$.

The paths $P$ and $Q$ are edge-disjoint, since we made sure that they use different edges incident with $s$, and among the remaining edges, $P$ uses only edges with one end-point in $S_1$ and the other one in $S_2$, while $Q$ uses only edges with both end-points in $S_1$, or both in $S_2$. Neither path uses the edge $st$.

\begin{figure}[t]
\centering
\subcaptionbox{\cref{sc:odd-not-halfing}\label{fig:pq-paths-odd}}{\includegraphics[page=1]{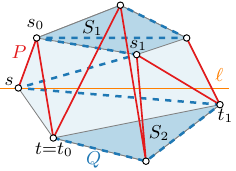}}
\hfil
\subcaptionbox{\cref{ssc:at-least-three}\label{fig:pq-paths-even-3}}{\includegraphics[page=3]{pq-paths.pdf}}
\hfil
\subcaptionbox{\cref{ssc:exactly-two}\label{fig:pq-paths-even}}{\includegraphics[page=2]{pq-paths.pdf}}
\caption{Illustration for \cref{sc:odd-not-halfing,ssc:at-least-three,ssc:exactly-two} of the proof of \cref{thm:pq-paths}.}
\label{fig:pq-paths}
\end{figure}

\subcase{sc:even-halfing}\textit{$|S|$ is even and $\overline{st}$ is a balancing line.} Then $\ell=\overline{st}$ defines a \bsp $S\setminus\{s,t\}=(S_1,S_2)$. For the illustrative figures, suppose that $\ell$ is horizontal, $s$ is the leftmost point of $\ch(S)\cap\ell$ and $S_1$ is above and $S_2$ below $\ell$. Since $t$ is also on $\pch(S)$, $t$ is the rightmost point of $\ch(S)\cap\ell$. Let $s_0t_0$ be the edge of $\ch(S_1\cup S_2)$ which intersects $\ell$ and whose intersection with $\ell$ is leftmost possible. Suppose without loss of generality that $s_0\in S_1$ and $t_0\in S_2$.

\subsubcase{ssc:at-least-three}\textit{$|S_2(s)\cup S_2(t)|\ge 3$.}
Let $Z$ be the zig-zag path for $S_1\cup S_2$ starting in point $s_0$. Denote its vertices by $z_0,z_1,\ldots,z_h$, with $z_0=s_0,z_2,\ldots,z_{h-1}\in S_1$, $z_1,z_3,\ldots,z_h\in S_2$, and the crossing points of $Z$ with the line $\ell$ by $x_1, x_2, \ldots, x_h$, with $x_i=z_{i-1}z_i\cap\ell$ for $i=1,2,\ldots,h$. We know that the crossing points $x_1,\ldots, x_h$ are ordered from left to right. Since $t$ is on $\pch(S)$, it lies to the right of $x_h$. The interior of the triangle $x_hz_ht$ does not contain any point of $S$, and hence the edge $z_ht$ does not intersect any edge of the zig-zag path $Z$. Therefore the path $P=s \circ ss_0 \circ Z \circ z_ht \circ t$ is a plane spanning path of~$S$.

The assumption $|S_2(s)\cup S_2(t)|\ge 3$ implies, via \cref{lem:bad2case}, that there exist $b\neq z_h, b\in S_2(t)$ and $a\in S_2(s), a\neq b$ such that $sa$ and $bt$ are non-crossing. \cref{lem:seestwo} implies that $s$ sees a point $c\neq s_0$ on $\pch(S_1)$. Now let $Q_2$ be a plane spanning path in $S_2$ starting at $b$ and ending at $a$ (guaranteed by \cref{lem:pq-path}) and $Q_1$ be a plane spanning path in $S_1$ starting at $c$. Then $Q=t \circ tb \circ Q_2 \circ as \circ s \circ sc \circ Q_1$ is a plane spanning path for $S$, and $P$ and $Q$ are edge-disjoint.

\subsubcase{ssc:other-at-least-three}\textit{$|S_1(s)\cup S_1(t)|\ge 3$.}
This case is symmetric to \cref{ssc:at-least-three}, we just start the zig-zag path $Z$ in $t_0$.

\subsubcase{ssc:exactly-two}\textit{$|S_1(s)\cup S_1(t)| = |S_2(s)\cup S_2(t)| = 2$.}
Consider again the path $P=s \circ ss_0 \circ s_0 \circ Z \circ z_h \circ z_ht \circ t$ as in \cref{ssc:at-least-three} and note that, due to the assumption that both $s$ and $t$ see the same two points on $\pch(S_1)$ (and $\pch(S_2)$), it follows that $z_1=t_0$ and that $z_{h-1}z_h$ is an edge of $\pch(S_1\cup S_2)$. Further, the edge $z_0z_h$ does not belong to $Z$ (since $h\ge 3$, and thus $z_h\neq z_1$). Let $Q_1$ be a plane spanning path for $S_1$ starting at $z_{h-1}$ and ending at $s_0$, and let $Q_2$ be a plane spanning path for $S_2$ starting at $z_h$ and ending at $t_0$ (the existence of such paths follows from \cref{lem:pq-path}). Then $Q=t \circ tz_{h-1} \circ Q_1 \circ s_0z_h \circ Q_2 \circ t_0s \circ s$ is a plane spanning path which is edge-disjoint with $P$. None of $P$ and $Q$ uses the edge $st$.

\ccase{c:equal}{$s=t$.}  Let $S'=S\setminus\{s\}$. Then $|S'|\ge 4$ and \cref{lem:seestwo} implies that $|S'(s)|\ge 2$. Let $a\neq b$ be two consecutive points on $\pch(S')$ which are seen by $s$. If $|S|\ge 6$, then $|S'|\ge 5$ and, by the already
proven \cref{c:unequal}, $S'$ contains two  edge-disjoint plane spanning paths $P'$ (starting at $a$) and $Q'$ (starting at $b$). It follows that $P=s \circ sa \circ P'$ and $Q=t \circ tb \circ Q'$ are edge-disjoint plane spanning paths for $S$ both starting in $s=t$; see \cref{fig:caseanalysis-ge6}.

For $|S|=5$, we have $|S'|=4$. It is easy to see that for any two consecutive points on %
$\pch(S')$, there exist two edge-disjoint plane spanning paths starting in these points; see \cref{fig:caseanalysis-4}. The paths starting at point $s$ are then constructed in the same way as in the case of $|S|\ge 6$.
\end{proof}

\begin{figure}[t]
\centering
\subcaptionbox{$|S|\ge 6$\label{fig:caseanalysis-ge6}}{\includegraphics[page=1]{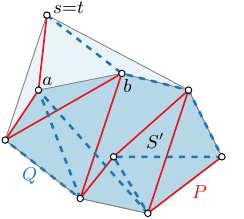}}
\hfil
\subcaptionbox{$|S|=5$\label{fig:caseanalysis-4}}{\includegraphics[page=2]{pq-caseanalysis.pdf}}
\caption{Illustration for \cref{c:equal} of the proof of \cref{thm:pq-paths}. }
\label{fig:caseanalysis}
\end{figure}

\section{Three edge-disjoint plane spanning paths}\label{sec:Main}

We first introduce a few technical notions. For two points $u,v$ in the plane, we denote by $(uv)^+$ the open halfplane to the right of the line
$\overline{uv}$, if the line is traversed in the way that $u$ precedes $v$.
The opposite open halfplane is denoted by $(uv)^-$.
 Note that $(uv)^-=(vu)^+$.
Let $Q$ be a convex polygon and let $u,v$ be two adjacent vertices of $Q$ such that
$Q\subseteq (uv)^+$.
Then we say that $v$ is the \emph{clockwise neighbor} of $u$ along (the boundary of) $Q$ and $u$ is the \emph{counterclockwise neighbor} of $v$ along (the boundary of) $Q$. We shall omit the words ``the boundary of'' when talking about the (counter-)clockwise neighbor along $Q$.
Let $(S_1,S_2)$ be a \bsp of $S$. The \emph{visibility graph} of the partition is
\[\vis(S_1,S_2)=(S,\{ab:a\in S_1, b\in S_2, {ab}\cap (\ch(S_1)\cup \ch(S_2))=\{a,b\} \}),\]
i.e., $ab\in E(\vis(S_1,S_2))$ if and only if $a\in S_1(b)$ and $b\in S_2(a)$; see \cref{fig:visibility}.
A path $v_1\circ v_1v_2\circ  \ldots \circ v_k$ of $\vis(S_1,S_2)$ is called \emph{switchable} if its edges $v_1v_2, v_2v_3, \ldots, v_{k-1}v_k$ cross the separating line of the partition $(S_1,S_2)$ in this order and for each $i=1,2,\ldots,k-2$, the interior of the triangle $v_i v_{i+1} v_{i+2}$ contains no point of $S$. Observe that every switchable path is non-crossing. We will show in Lemma~\ref{lem:bigN} that if a zig-zag path $Z$ contains a switchable path $a\circ ab \circ b \circ bc \circ c \circ cd \circ d$ of length 3 as a subpath, it can be modified to a plane path $\ldots \circ a \circ ac \circ c \circ cb \circ b \circ bd \circ d \ldots$ which allows spanning paths in the two classes of the \bsp\ to be concatenated via the edges $ab$ and $cd$.%

\begin{figure}[t]
\centering
{\includegraphics{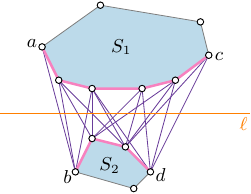}}
\caption{Illustration for the definition of the visibility  graph of a \bsp. The edges of the visibility graph are drawn purple.}
\label{fig:visibility}
\end{figure}

The \emph{$n$-wheel configuration} $W_n$ of points  (in general position) is a set of $n-1$ points in convex position, augmented with one point lying inside the convex hull of these $n-1$ points in such a position that every line that passes through the augmenting point and any other point is a balancing line of $W_n$. This point configuration plays an important role in the proof below, and we need to show that it contains three edge-disjoint plane spanning paths by an ad hoc construction, at least for the case of $n$ even. This has already been sketched by Aichholzer et al,~\cite{ahkklpsw-ppstpcgg-17}. %

\begin{proposition}
\label{prop:even-wheel}
For even $n\geq6$, the maximum number of edge-disjoint plane spanning paths in the wheel configuration $W_n$ is $\frac{n}{2}-1$.
\end{proposition}

\begin{figure}[b]
\centering
\subcaptionbox{$P_1$\label{fig:evenwheel-1}}{\includegraphics[page=6]{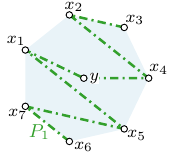}}
\hfil
\subcaptionbox{$P_2$\label{fig:evenwheel-2}}{\includegraphics[page=7]{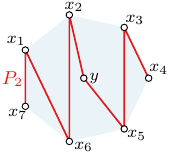}}
\hfil
\subcaptionbox{$P_3$\label{fig:evenwheel-3}}{\includegraphics[page=8]{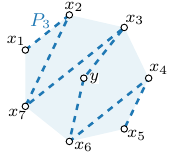}}
\hfil
\subcaptionbox{All three paths\label{fig:evenwheel-4}}{\includegraphics[page=9]{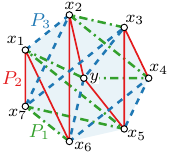}}
\caption{Illustration for the proof of \cref{prop:even-wheel} for $W_8$.}
\label{fig:evenwheel}
\end{figure}

\begin{proof}
Let $y$ be the augmenting point lying inside $\ch(W_n)$ and let $x_1,x_2,
\ldots,x_{n-1}$ be the points on $\pch(W_n)$ listed in the
clockwise order; see \cref{fig:evenwheel}. First construct $\frac{n}{2}-1$ plane paths
$P'_k = \ldots %
        \circ x_{k-1}
        \circ x_{k-1}x_{k+\frac{n}{2}} \circ x_{k+\frac{n}{2}}
        \circ x_{k+\frac{n}{2}}x_k \circ x_k
        \circ x_kx_{k+\frac{n}{2}-1} \circ x_{k+\frac{n}{2}-1}
        \circ x_{k+\frac{n}{2}-1}x_{k+1} \circ x_{k+1}
        \circ \ldots$
for $k=1,2,\ldots,\frac{n}{2}-1$, with counting in subscripts modulo $n-1$. Each of
the paths is non-crossing and spans the points $W_n\setminus\{y\}$. They
are pairwise edge-disjoint, because for every $k$, if $x_i$ and $x_j$ are
consecutive on $P'_k$, then $i+j=2k+\frac{n}{2}-1$ or $2k+\frac{n}{2}$. Then,
again for each $k$, create a path $P_k$ from $P'_k$ by replacing the edge
$x_kx_{k+\frac{n}{2}-1}$ by the path  $ x_ky \circ y \circ yx_{k+\frac{n}{2}-1}$.
Again, the edges of these subpaths are private for them, and thus the paths
$P_k, k=1,2,\ldots,\frac{n}{2}-1$ are edge-disjoint. And each of them is
non-crossing, because for each $k$, $\overline{x_ky}$ is a balancing line. And they
clearly span all points of $W_n$.

The fact that $W_n$ does not admit $\frac{n}{2}$ edge-disjoint
plane spanning paths for $n\geq 6$ is implied by the work of Biniaz et al.~\cite[Theorem 4]{DBLP:journals/dmtcs/BiniazBMS15} who prove that the maximum number of edge-disjoint perfect matchings of $W_n$ is $\frac{n}{2} -1$.
\end{proof}

Our later proof of \cref{thm:main} is based on the following structural result. 
Recall that a vertex is \emph{bridged} if it is incident to one of the two edges
of $\ch(S)$ that is between endpoints of $S_1$ and $S_2$.

\begin{theorem}\label{thm:crucial}
Let $S$ be a set of $n\ge 5$ points in general position in the plane. Then at least one of the following holds true
\begin{enumerate}
\item\label[condition]{thm:crucial-1} $S$ has a \bsp $(S_1,S_2)$ such that $\vis(S_1,S_2)$ contains two crossing edges, or
\item\label[condition]{thm:crucial-2} $S$ has a \bsp $(S_1,S_2)$ such that $\vis(S_1,S_2)$ contains a switchable path of length 3 and a bridged vertex not included in the path which is  incident with at
least 2 edges of $\vis(S_1,S_2)$, or
\item\label[condition]{thm:crucial-3} $n$ is even and $S$ is the wheel configuration $W_n$.
\end{enumerate}
\end{theorem}

\begin{proof}
{\bf Case 1: $n$ is odd}.
A line $\overline{xy}$ passing through two points $x,y\in S$ is an \emph{almost-balancing line} of $S$,
if exactly $(n-1)/2$ points of $S\setminus\{x,y\}$ lie on one of the sides of $\overline{xy}$ and
the remaining $(n-3)/2$ points of $S\setminus\{x,y\}$ lie on the other side of $\overline{xy}$.
We fix an extreme point $u$ of $S$. Let $\overline{ua}$, $\overline{ub}$ be the two almost-balancing lines passing through $u$; see \cref{fig:crucial-partition}.
Suppose $b\in (ua)^+$.
Note that the interior of the convex wedge bounded by the two rays emanating from $u$, one passing through $a$ and the other passing through $b$, contains no point of $S$.
Each of the lines $\overline{ua}$ and $\overline{ub}$ partitions the set $S\setminus\{u,a,b\}$ into two sets $A$ and $B$ of equal size
$(n-3)/2$, such that $A$ lies to the left of the lines $\overline{ua}$ and $\overline{ub}$ and $B$ lies to the right of them.
The line $\overline{ab}$ partitions $A$ into $A_1$ and $A_2$, and $B$ into $B_1$ and $B_2$, such that $A_1$ and $B_1$
lie in $({ab})^+$, and  $A_2$ and $B_2$ lie in $({ab})^-$.

\begin{figure}[t]
\centering
\subcaptionbox{$A_1,A_2,B_1,B_2$\label{fig:crucial-partition}}{\includegraphics[page=1]{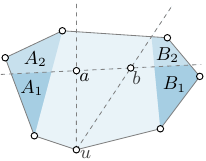}}
\hfill
\subcaptionbox{$A_1\neq\emptyset$\\ $\rightarrow$ \cref{thm:crucial-1}\label{fig:crucial-nonempty}}{\includegraphics[page=2]{crucial.pdf}}
\hfill
\subcaptionbox{$A_1=B_1=\emptyset$\\ $\rightarrow$ \cref{thm:crucial-1}\label{fig:crucial-empty}}{\includegraphics[page=3]{crucial.pdf}}
\caption{Illustration for Case~1 of the proof of \cref{thm:crucial}.}
\label{fig:crucial}
\end{figure}

Suppose first that $A_1\neq\emptyset$, and let $a_{\circlearrowright}$ be the clockwise neighbor of $a$ along $\ch(A\cup\{a\})$; see \cref{fig:crucial-nonempty}.
Consider the \bsp $(S_1,S_2)$ with $S_1=A\cup\{a\}$ and $S_2=B\cup\{b,u\}$ of $S$. Its visibility graph $\vis(S_1,S_2)$ contains the crossing
edges $au$ and $a_{\circlearrowright}b$, which %
proves the result in this case.

If $B_1\ne\emptyset$, then we can analogously find a crossing pair of edges in $\vis(A\cup\{a,u\},B\cup\{b\})$.
Thus, we may further assume that $A_1=B_1=\emptyset$. Then we have $|A_2|=|B_2|=(n-3)/2>0$.
Let $a_\circlearrowleft\in A_2$ be the counterclockwise neighbor of $a$ along $\ch(A\cup\{a\})$, and let $b_\circlearrowright$ be
the clockwise neighbor of $b$ along $\ch(B\cup\{u,b\})$. Then $a_\circlearrowleft b$, $ab_\circlearrowright$ is a crossing pair of edges
of $\vis(A\cup\{a\},B\cup\{b,u\})$. %

\smallskip\noindent
{\bf Case 2: $n$ is even}.
A line passing through two points $x,y\in S$ is a \emph{halving line} of $S$
if exactly $\frac{n-2}{2}$ points of $S$ lie on each of its two sides.
If $\overline{xy}$ is a halving line of $S$, where $x,y\in S$, then the segment $xy$ is called
a \emph{halving segment} of $S$.

\begin{claim}\label{clm:halving}
Let $uv$ be a halving segment of a set $S$ of $n$ points in general position in the plane such that $u\in\pch(S)$.
Then there is another halving segment $pq$ of $S$ such that the following three conditions hold:
\begin{enumerate*}[label=(\arabic*)]
\item an unbounded part of the ray emanating from $p$ and passing through $q$ lies in $(uv)^+$;
\item no point of $S\setminus\{u,v,p,q\}$ lies in the double-wedge $((uv)^+\cap (pq)^-)\cup((uv)^-\cap (pq)^+)$;
\item $p=v$, or the two open segments $uv$ and $pq$ cross.
\end{enumerate*}
\end{claim}

\begin{proof}
Suppose without loss of generality that $uv$ is a vertical line and $u$ lies below $v$; see \cref{fig:halving}. Let $X:=\ch(S\cap (uv)^-)$
and $Y:=\ch(S\cap (uv)^+)$. Let $x\in X$ and $y\in Y$ be the extreme points of $X$ and $Y$, respectively, such that
$\overline{xy}$ avoids the interiors of $X$ and $Y$, the set $X$ lies above $\overline{xy}$, and
$Y$ lies below $\overline{xy}$.

If the open segments $uv$ and $xy$ cross, then
$p=x$ and $q=y$ yields the %
claim; see \cref{fig:halving-cross}.
Otherwise, $xy$ crosses the line $\overline{uv}$ above $v$; see \cref{fig:halving-nocross}.
Let $p=v$ and $q\in Y$ be the extreme point of $Y$ seen from $p$ as the highest point of $Y$. (That is, all
points of $Y\setminus\{q\}$ lie below the line $\overline{pq}$.) Then $pq$ is a halving segment
and yields the claim.
\end{proof}

\begin{figure}[t]
\centering
\subcaptionbox{Segments $uv$ and $xy$ cross\label{fig:halving-cross}}[.47\textwidth]{\includegraphics[page=1]{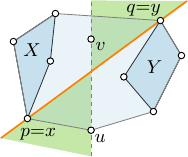}}
\hfill
\subcaptionbox{Segments $uv$ and $xy$ do not cross\label{fig:halving-nocross}}[.47\textwidth]{\includegraphics[page=2]{halving.pdf}}
\caption{Illustration for \cref{clm:halving}.}
\label{fig:halving}
\end{figure}

Fix an extreme point $u$ of $S$. Let $uv$ be the unique halving segment incident to $u$.
Let $A:=S\cap (uv)^-$ and $B:=S\cap (uv)^+$. Since $uv$ is a halving segment, we have $|A|=|B|=\frac{n-2}{2}$.
 Let $pq$ be the halving segment guaranteed by %
\cref{clm:halving}. %

If the halving segments $uv$ and $pq$  cross each other, then they both belong to
the visibility graph $\vis(A',B')$,
where $A':=A\cup\{v\}$ and $B':=B\cup\{u\}$, and \cref{thm:crucial-1} applies; see \cref{fig:crucial-cross}.
Thus we may assume that  $p=v$ and there are no points of $S$
inside the double wedge $((uv)^+\cap (vq)^-)\cup((uv)^-\cap (vq)^+)$.
By a mirror argument, we can also assume that there is a point $r\in A$ such that there are no points of $S$
inside the double wedge $((uv)^-\cap (vr)^+)\cup((uv)^+\cap (vr)^-)$.

The line $\overline{qr}$ partitions $A\setminus\{r\}$ ($B\setminus\{q\}$, respectively) into two sets
$A_1$ and $A_2$ ($B_1$ and $B_2$, respectively) such that $A_1$, $B_1$ lie in $(qr)^-$ and
$A_2$, $B_2$ lie in $(qr)^+$.

We now distinguish three subcases. Consider first the case $A_2\neq\emptyset$ and $B_2\neq\emptyset$; see \cref{fig:crucial-bothempty}.
Then the counterclockwise neighbor $r_\circlearrowleft $ of $r$ along $\ch(A\cup\{u\})$ lies in $A_2$.
Similarly, the clockwise neighbor $q_\circlearrowright$ of $q$ along $\ch(B\cup\{v\})$ lies in $B_2$.
It follows that the edges $qr_\circlearrowleft $ and $rq_\circlearrowright$ form a crossing pair in $\vis((A\cup\{u\}),(B\cup\{v\}))$, and \cref{thm:crucial-1} applies.

\begin{figure}[t]
\centering
\subcaptionbox{$uv$ and $pq$ cross\\ $\rightarrow$ \cref{thm:crucial-1}\label{fig:crucial-cross}}{\includegraphics[page=4]{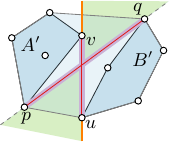}}
\hfill
\subcaptionbox{$A_2\neq \emptyset$ and $B_2\neq\emptyset$\\ $\rightarrow$ \cref{thm:crucial-1}\label{fig:crucial-bothempty}}{\includegraphics[page=5]{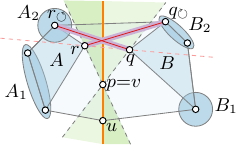}}
\hfill
\subcaptionbox{$A_2\neq\emptyset$ and $B_2=\emptyset$\\ $\rightarrow$ \cref{thm:crucial-2}\label{fig:crucial-oneempty}}{\includegraphics[page=6]{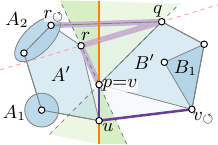}}

\medskip

\subcaptionbox{$A_2=B_2=\emptyset$\label{fig:crucial-bothempty-1}}{\includegraphics[page=7]{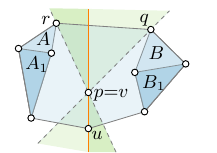}}
\hfill
\subcaptionbox{$A_2=B_2=\emptyset$, step 1\label{fig:crucial-bothempty-2}}{\includegraphics[page=8]{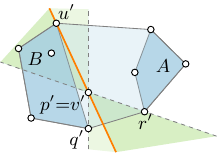}}
\hfill
\subcaptionbox{$A_2=B_2=\emptyset$, step~2\\ $\rightarrow$ \cref{thm:crucial-1}\label{fig:crucial-bothempty-3}}{\includegraphics[page=9]{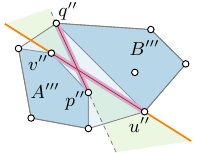}}

\caption{Illustration for Case 2 of \cref{thm:crucial}.}
\label{fig:case2}
\end{figure}

Consider now that one of the sets $A_2$ and $B_2$ is empty and the other one is non-empty; see \cref{fig:crucial-oneempty}.
By symmetry, we may assume that $A_2\neq\emptyset$ and $B_2=\emptyset$.
Then $B_1\neq\emptyset$.
We consider the \bsp $(A',B')$, where $A':=A\cup\{u\}$ and $B':=B\cup\{v\}$. Let $r_\circlearrowleft $ be the counterclockwise neighbor of $r$
along $\ch(A')$. Since $A_2\neq\emptyset$, $r_\circlearrowleft $ lies in $A_2$.
Then (1) $u$ is a bridged vertex for the partition $(A',B')$, (2)
$u$ is incident with at least 2 edges of $\vis(A',B')$ -- the edge $uv$ and the
edge $uv_{\circlearrowleft}$ for the counter-clockwise neighbor $v_{\circlearrowleft}$ of $v$ along
$\ch(B')$, and
(3) %
 $v\circ vr\circ r\circ rq\circ q\circ qr_\circlearrowleft \circ r_\circlearrowleft $ is a switchable path in $\vis(A',B')$, and \cref{thm:crucial-2} applies.

Finally, consider that  $A_2=\emptyset$ and $B_2=\emptyset$; see \cref{fig:crucial-bothempty-1}. Then $q$ and $r$ are neighbors along $\ch(S)$, and we again consider the whole analysis which started with fixing an extreme point of $S$ but
now we fix the point $u':=r$ instead of $u$; see \cref{fig:crucial-bothempty-2}. Either we find a \bsp satisfying  \cref{thm:crucial-1} or \cref{thm:crucial-2},
or we find two neighbors $q'$ and $r'$ along $\ch(S)$. In the first case we are done.
In the latter case, the point $q'$ is actually equal to $u$ and it is clockwise of $r'$
along $\ch(S)$. 

We then again consider the analysis which started with fixing an extreme point $S$, but
now we fix the point $u'':=r'$ instead of $u$; see \cref{fig:crucial-bothempty-3}. 
We continue with this process, fixing point $u^{(k)}=r^{(k-1)}$ in step $k$, until
we find a \bsp satisfying condition 1) or 2) of the theorem.
Note that, in this procedure, the line $uv$ keeps rotating clockwise through the convex hull of $S$. As long as there are at least two points inside the convex hull, we are guaranteed to find one of the above cases at some point: it can only rotate infinitely if we have $v^{(k)}=v^{(k-1)}$ and $p^{(k)}=p^{(k-1)}$ at every step. 
After $n/2$ repetitions of our procedure, the line $uv$ will have made a $180^\circ$ rotation, so it must have passed another interior point somewhere, which would change $v$; unless $S$ is the wheel configuration $W_n$. %
\end{proof}

Let $Z$ be a zig-zag path for a partition $(S_1,S_2)$ of $S$. For two vertices $a,b\in S$ , we call the segment $ab$ a \emph{free edge} (with respect to $Z$) if $ab\in E(\vis(S_1,S_2))$ and $ab\not\in E(Z)$.

\begin{lemma}\label{lem:2freeedges}
Let $|S|\ge 10$ and let $Z$ be a zig-zag path for a \bsp {} $(S_1,S_2)$ of $S$ which leaves at least two free edges. Then $S$ allows three edge-disjoint plane spanning paths.
\end{lemma}

\begin{figure}[t]
\centering
\subcaptionbox{$Z$ with two free edges $ab, cd$\label{fig:2freeedges-Z}}{\includegraphics[page=1]{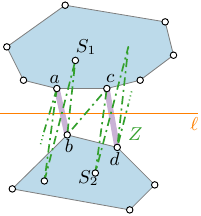}}
\hfill
\subcaptionbox{$P$ and $Q$\label{fig:2freeedges-PQ}}{\includegraphics[page=2]{2freeedges}}
\hfill
\subcaptionbox{The three paths\label{fig:2freeedges-all3}}{\includegraphics[page=3]{2freeedges}}
\caption{Illustration for the proof of \cref{lem:2freeedges}}
\label{fig:2freeedges}
\end{figure}

\begin{proof}
Let $ab$ and $cd$ be two free edges with respect to a zig-zag path $Z$ and a
\bsp $(S_1,S_2)$; see \cref{fig:2freeedges-Z}. Since $|S|\ge 10$, we have $|S_1|\ge 5$ and $|S_2|\ge 5$.
Suppose $a,c\in \pch(S_1)$ and $b,d\in \pch(S_2)$. We may have $a=c$ or $b=d$, but not
both.
Let $P_1$ and $Q_1$ be edge-disjoint plane
spanning paths for $S_1$, with $P_1$ starting at $a$ and $Q_1$ starting at $c$.
Similarly, let $P_2$ and $Q_2$ be edge-disjoint plane spanning paths for
$S_2$, with $P_2$ starting at $b$ and $Q_2$ starting at $d$. The existence of
such paths is guaranteed by \cref{thm:pq-paths}. Then $
P=P_2^{-1} \circ ba \circ P_1$ and $Q=Q_2^{-1} \circ dc \circ Q_1$ are edge-disjoint plane
spanning paths for $S$; see \cref{fig:2freeedges-PQ}. Each of them is plane because the edge $ab$
($cd$, respectively) contains no point in the interior of $\ch(S_1)$ (of
$\ch(S_2)$, respectively). Both of them are edge-disjoint with $Z$, because the only
two edges of them that are incident with vertices from both $S_1$ and $S_2$
are $ab$ and $cd$, and these are by assumption free w.r.t.\ $Z$; see \cref{fig:2freeedges-all3}.
\end{proof}

The following lemmas can be proved similarly with \cref{lem:seestwo,lem:claimbad2case}.

\begin{lemma}
\label{lem:bigN}
Let $|S|\ge 10$ and let $Z$ be a zig-zag path for a \bsp $(S_1,S_2)$ of $S$ which contains all three edges of a switchable path of length~3 in $\vis(S_1,S_2)$. Then $S$ contains three edge-disjoint plane spanning paths.
\end{lemma}

\begin{figure}[t]
\centering
\subcaptionbox{$Z$\label{fig:bigN-Z}}{\includegraphics[page=1]{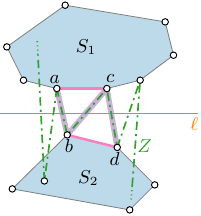}}
\hfill
\subcaptionbox{$Z'$\label{fig:bigN-Z'}}{\includegraphics[page=2]{bigN}}
\hfill
\subcaptionbox{$P$ and $Q$\label{fig:bigN-PQ}}{\includegraphics[page=3]{bigN}}
\caption{Illustration for the proof of \cref{lem:bigN}}
\label{fig:bigN}
\end{figure}

\begin{proof}
Note that $|S_i|\ge 5$ for $i=1,2$.
Let $ab,bc,cd$ be the edges of a switchable path $P\subseteq Z$, and let they occur on $Z$ in this order, i.e., $Z=Z_1 \circ ab \circ b \circ bc \circ c \circ cd \circ Z_2$. Since $a,c$ ($b,d$) are neighbors along $\pch(S_1)$ ($\pch(S_2)$, respectively), the  edges $ac$ and $bd$ are not crossed by any edge of $Z$, and hence $Z'=Z_1 \circ ac \circ c \circ cb \circ b \circ bd \circ Z_2$ is a plane spanning path for $S$; see \cref{fig:bigN-Z'}. Now \cref{thm:pq-paths} implies that there exist edge-disjoint plane spanning paths  $P_1, Q_1$ of $S_1$, $P_1$ starting at $a$ and $Q_1$  starting at $c$, none of them containing the edge $ac$.  Similarly, there exist  edge-disjoint plane spanning paths $P_2, Q_2$ for $S_2$, $P_2$ starting at $b$ and $Q_2$ starting at $d$, none of them containing the edge $bd$; see \cref{fig:bigN-PQ}. It follows that $P=P^{-1}_1 \circ ab \circ P_2$ and $Q=Q^{-1}_1 \circ cd \circ Q_2$ are  edge-disjoint plane spanning paths of $S$, and they are both edge-disjoint with $Z'$, since $Z'\cap \{ab,cd\}=\emptyset$ and $(P\cup Q)\cap \{ac,bd\}=\emptyset$; see \cref{fig:bigN-PQ}.
\end{proof}

\begin{lemma}
\label{lem:diamond}
Let $|S|\ge 10$ and let $(S_1,S_2)$ be a \bsp of $S$ such that $\vis(S_1,S_2)$ contains two crossing edges. Then $S$ contains three edge-disjoint plane spanning paths.
\end{lemma}

\begin{proof}
Note that $|S_i|\ge 5$ for $i=1,2$. We first argue that $\vis(S_1,S_2)$ contains two crossing edges $ad$ and $bc$, with $a,c\in S_1$ and $b,d\in S_2$, such that $a$ and $c$ are consecutive on $\pch(S_1)$ and $b$ and $d$ are consecutive on $\pch(S_2)$; see \cref{fig:diamond-abcd}.
To see this, suppose $a'd', b'c'$ are crossing edges of $\vis(S_1,S_2)$, with $a',c'\in S_1$ and $b',d'\in S_2$. \cref{lem:claimbad2case} implies that $a'b',c'd'\in E(\vis(S_1,S_2))$.
\cref{lem:seestwo} states that every point on the boundary path of $\ch(S_1)$ between $a'$ and $c'$ sees both $b'$ and $d'$, and any point on $\pch(S_2)$ between $b'$ and $d'$ sees both $a'$ and $c'$, and consequently, these points induce a complete bipartite subgraph of $\vis(S_1,S_2)$. Hence it suffices to take any $c=c'$ and as $a$ its neighbor along $\pch(S_1)$ in the direction to $a'$, and $d=d'$ and as $b$ its neighbor along $\pch(S_2)$ in the direction to $b'$.

\begin{figure}[t]
\centering
\subcaptionbox{crossing edges between consecutive vertices exist\label{fig:diamond-abcd}}{\includegraphics[page=1]{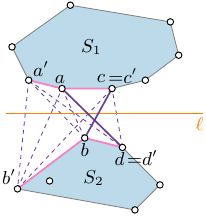}}
\hfill
\subcaptionbox{two free edges\\ $\rightarrow$ \cref{lem:2freeedges}\label{fig:diamond-twofree}}{\includegraphics[page=2]{diamond}}
\hfill
\subcaptionbox{switchable path of length three\\ $\rightarrow$ \cref{lem:bigN}\label{fig:diamond-switchable}}{\includegraphics[page=3]{diamond}}
\caption{Illustration for the proof of \cref{lem:diamond}}
\label{fig:diamond}
\end{figure}

Now consider the four points $a,b,c,d$, and consider a zig-zag path $Z$ with respect to the partition $(S_1,S_2)$. If $Z$ contains at most 2 of the edges $ab,ad,bc,cd$, $\vis(S_1,S_2)$ contains 2 free edges w.r.t. $Z$ and $S$ contains 3 edge-disjoint plane spanning paths according to \cref{lem:2freeedges}; see \cref{fig:diamond-twofree}. Obviously, $Z$ cannot contain all four of these edges, since $Z$ is non-crossing. In the remaining case, when $Z$ contains exactly three of the edges $ab,bc,ad,cd$, suppose w.l.o.g. that $Z$ contains $ab,bc,cd$. Then \cref{lem:bigN} implies that $S$ allows three edge-disjoint plane spanning paths; see \cref{fig:diamond-switchable}.
\end{proof}

\begin{lemma}
\label{lem:bigN+}
Let $|S|\ge 10$ and let $(S_1,S_2)$ be a \bsp of $S$ with $|S_1|\ge |S_2|$ such that $\vis(S_1,S_2)$ contains a switchable path of length~3 and a bridged vertex in $S_1$ which does not belong to the switchable path  and which is incident with at least two edges of
$\vis(S_1,S_2)$. Then $S$ contains three  edge-disjoint plane spanning paths.
\end{lemma}

\begin{proof}
Let $u\in S_1$ be a bridged vertex of $S_1$ which does not belong to a switchable path $P$ and is incident to at least two edges of $\vis(S_1,S_2)$; see \cref{fig:bigNPlus}. Let $Z$ be a zig-zag path starting in point $u$ (since $u$ is incident with a bridge of the partition $(S_1,S_2)$, such a path always exists as argued in the sketch of proof of \cref{lem:zigzag}). Then the degree of $u$ in $Z$ is 1,  and since $u$ is incident with at least 2 edges of $\vis(S_1,S_2)$
by the assumption, $u$ is incident with at least one free edge w.r.t. $Z$.
If $Z$ misses at least one of the edges of $P$, it leaves at least two free edges and $S$ contains three edge-disjoint plane spanning paths by \cref{lem:2freeedges}; see \cref{fig:bigNPlus-miss}. Otherwise, $Z$ contains all three edges of the switchable path $P$, and three edge-disjoint plane spanning paths exist according to \cref{lem:bigN}; see \cref{fig:bigNPlus-hit}.
\end{proof}

\begin{figure}[t]
\centering
\subcaptionbox{$Z$ misses one edge of $P$ $\rightarrow$ \cref{lem:2freeedges}\label{fig:bigNPlus-miss}}[.47\textwidth]{\includegraphics[page=1]{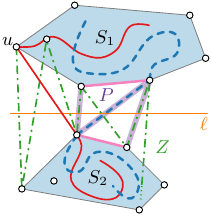}}
\hfil
\subcaptionbox{$Z$ contains $P$ $\rightarrow$ \cref{lem:bigN}\label{fig:bigNPlus-hit}}[.47\textwidth]{\includegraphics[page=2]{bigNPlus}}
\caption{Illustration for the proof of \cref{lem:bigN+}}
\label{fig:bigNPlus}
\end{figure}

We are now ready to prove our main result \cref{thm:main}.

\ThmMain*
\begin{proof}
Given a set of points $S$, apply \cref{thm:crucial}. If $S$ allows a \bsp with two crossing edges in its visibility graph, $S$ has three edge-disjoint plane spanning paths according to \cref{lem:diamond}. If $S$ allows a \bsp whose visibility graph contains a switchable path of length three and a bridged vertex, then $S$ has three edge-disjoint plane spanning paths according to \cref{lem:bigN+}. If none of these cases apply, then by \cref{thm:crucial} $S$ is the wheel configuration $W_n$ and $n$ is even, in which case $S$ has $n/2-1\ge 4$ edge-disjoint plane spanning paths according to \cref{prop:even-wheel}.
\end{proof}

\section{Upper Bound}\label{se:Upper}

It is immediate to see that any set of $n$ points cannot have more than $\lfloor \frac{n}{2} \rfloor$ edge-disjoint spanning paths.
In this section we give a linear upper bound on the number of edge-disjoint plane spanning paths such that the multiplicative factor of the bound is smaller than $\frac{1}{2}$. Our argument extends to paths a similar result about perfect matchings by Biniaz et al.~\cite{DBLP:journals/dmtcs/BiniazBMS15}.

\begin{figure}[t]
\centering
\subcaptionbox{The construction with $n=9$ and $k=2$.\label{fig:kwheel-1}}[.47\textwidth]{\includegraphics[page=22]{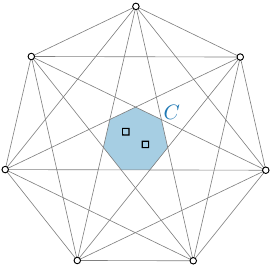}}
\hfil
\subcaptionbox{Construction of $e_i, e_j$, and $\Pi_0$.\label{fig:kwheel-2}}[.47\textwidth]{\includegraphics[page=23]{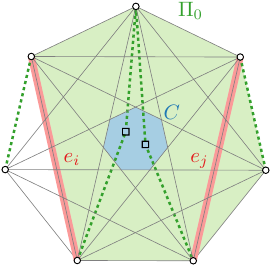}}

\medskip

\subcaptionbox{The construction with $n=10$ and $k=2$.\label{fig:kwheel-3}}[.47\textwidth]{\includegraphics[page=24]{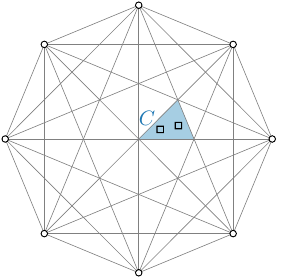}}
\hfil
\subcaptionbox{Construction of $e_i, e_j$, and $\Pi_0$.\label{fig:kwheel-4}}[.47\textwidth]{\includegraphics[page=25]{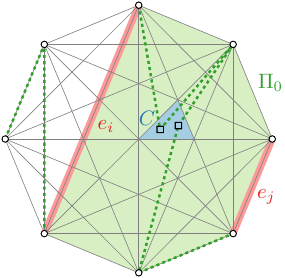}}
\caption{Illustration of the proof of \cref{th:upper}. Vertices in $S_1$ are drawn as circles and vertices in $S_2$ are drawn as squares. The region $C$ is shaded.}
\label{fig:kwheel}
\end{figure}
\begin{theorem}\label{th:upper}
For any $n \geq 6$, there exists a set of points $S$ such that the maximum number of edge-disjoint plane spanning paths is at most $\lceil \frac{n}{3} \rceil$.
\end{theorem}
\begin{proof}
Let $n \geq 6$ be an integer and let $k= \lfloor \frac{n}{3} \rfloor -1$. We consider a set $S$ of $n$ points consisting of a subset $S_1$ and a subset $S_2$. The set $S_1$ consists of the vertices of a regular $(n-k)$-gon $\Pi$; let $G$ be the complete geometric graph whose vertices are the points of $S_1$. The set $S_2$ consists of $k$ distinct points placed in the interior of a cell $C$ defined by the edges of $G$ and chosen as follows. If $|S_1| = n-k$ is odd, then $C$ is the cell containing the center of polygon $\Pi$; see \cref{fig:kwheel-1}. Otherwise, $C$ is any of the cells of $G$ whose boundary contains the center of polygon $\Pi$; see \cref{fig:kwheel-3}. 

We now prove that any plane spanning path of $S$ contains at least two edges of the convex hull of $S$. 
Let $e_1, e_2, \ldots e_{n-1}$ be the edges of a plane spanning path $P$ of $S$, as they appear in this order when traversing $P$ from one endpoint to the other. Since $|S_1| > |S_2| + 1$, path $P$ must contain at least one edge connecting two points of $S_1$. Let $e_i=uv$ be a longest such edge of $P$ ($1 \leq i \leq n-1$); see \cref{fig:kwheel-2,fig:kwheel-4}. Edge $e_i$ divides the interior of $\Pi$ into two regions. One of them, denoted as $\Pi_0$, contains $|S_2|$. By construction, the boundary of $\Pi_0$ contains at least $\frac{|S_1|}{2}-1$ points of $S_1 \setminus\{u,v\}$. Since $\frac{|S_1|}{2}-1 > |S_2|$, it follows that there is another edge of $P$, say $e_j$ ($j\neq i$), whose endpoints  both belong to $S_1 \cap \Pi_0$. By the choice of $e_i$, all points of $S_2$ lie between $e_i$ and $e_j$. It follows that $e_1$ and $e_{n-1}$ are edges of the convex hull of $S$.

Since we have shown that any plane spanning path of $S$ contains at least two edges of the convex hull of $S$, there cannot be more than $\lfloor \frac{|S_1|}{2} \rfloor = \lceil \frac{n}{3} \rceil$ edge-disjoint plane spanning paths of $S$. 
\end{proof}

\section{Conclusion}\label{se:conclusion}

In this paper, we showed that every set of at least 10 points in general position admits three edge-disjoint plane spanning paths. While we mostly focused on the combinatorial part, it is easy to see that our constructive arguments give rise to a polynomial time algorithm.
We note that it's a simple exercise to verify that the 6-wheel configuration does not contain three edge-disjoint spanning paths. On the other hand, it was verified by a computer program that all sets of 7, 8, or 9 points contain three edge-disjoint spanning paths~\cite{Scheucher}.

Of course, reducing the gap highlighted by \cref{thm:main} and \cref{th:upper} appears to be the most interesting question. Some intermediate steps that could be leading to this goal are listed below.

Can \cref{thm:pq-paths} be strengthened?
Does any set of $n$ points (for large enough $n$) in general position contain, for any choice of two distinct points $s,t$ (not necessarily lying on the boundary of the convex hull of the set),  edge-disjoint plane spanning paths starting in these points and not containing the edge $st$?

Let us mention in this connection that \cref{thm:main} cannot be strengthened in the way of \cref{thm:pq-paths}. If the points of $S$ are in convex position and the starting points of the three paths are prescribed to be the same point of $S$, then three  edge-disjoint plane spanning paths do not exist (for a convex position, every path must start with an edge of $\pch(S)$, and for a single point, there are only two such edges). The question is currently open to us if the starting points are required to be distinct.

\section{Acknowledgments}

The second and fourth authors gratefully acknowledge the support of Czech Science Foundation through research grant GA\v{C}R 23-04949X. The work of the third author is partially supported by "(i) MUR PRIN Proj. 2022TS4Y3N - ``EXPAND: scalable algorithms for EXPloratory Analyses of heterogeneous and dynamic Networked Data''; (ii) MUR PRIN Proj. 2022ME9Z78 - ``NextGRAAL: Next-generation algorithms for constrained GRAph visuALization''. All authors acknowledge the working atmosphere of Homonolo meetings where the research was initiated and part of the results were obtained, as well as of Bertinoro Workshops on Graph Drawing, during which we could meet and informally work on the project. Our special thanks go to Manfred Scheucher whose experimental results encouraged us to keep working on the problem in the time when all hopes for a solution seemed far out of sight.

\bibliographystyle{elsarticle-harv}
\bibliography{biblio,bib/knizky,bib/nakryti,bib/sborniky,bib/litRN}
\end{document}